\documentclass[sigconf,nonacm]{acmart}
\usepackage{latexsym}
\usepackage{soul}
\usepackage{graphicx}
\usepackage{multirow}
\usepackage{xcolor}
\usepackage{tabularx}
\usepackage{booktabs}
\usepackage{mathtools}
\usepackage{subcaption}
\usepackage[linesnumbered,ruled]{algorithm2e}
\usepackage{makecell}
\usepackage{bbm}

\newcommand{\ctext}[3][RGB]{%
  \begingroup
  \definecolor{hlcolor}{#1}{#2}\sethlcolor{hlcolor}%
  \hl{#3}%
  \endgroup
}

\newcommand{\cpu}[1]{\ctext[RGB]{198,219,239}{#1}}
\newcommand{\gpu}[1]{\ctext[RGB]{168,221,181}{#1}}
\newcommand{\hide}[1]{\textcolor{gray}{#1}}

\newcommand{\trecdldf}{\textsc{Doc'19}}
\newcommand{\trecdlds}{\textsc{Doc'20}}
\newcommand{\trecdlp}{\textsc{Passage'19}}

\newcommand{\sparseretrieval}{\textsc{Sparse Retrieval}}
\newcommand{\denseretrieval}{\textsc{Dense Retrieval}}
\newcommand{\hybrid}{\textsc{Hybrid Retrieval}}
\newcommand{\reranking}{\textsc{Re-Ranking}}
\newcommand{\interpolatedreranking}{\textsc{Interpolation}}
\newcommand{\fastforward}{\textsc{Fast-Forward}}

\newcommand{\bm}{\textsc{BM25}}
\newcommand{\deepct}{\textsc{DEEP-CT}}
\newcommand{\tct}{\textsc{TCT-ColBERT}}
\newcommand{\ance}{\textsc{ANCE}}
\newcommand{\clear}{\textsc{CLEAR}}
\newcommand{\cls}{\textsc{BERT-CLS}}

\newcommand{\pyserini}{\textsc{Pyserini}}

\copyrightyear{2022}
\acmYear{2022}
\setcopyright{acmcopyright}\acmConference[WWW '22]{Proceedings of the ACM Web Conference 2022}{April 25--29, 2022}{Virtual Event, Lyon, France}
\acmBooktitle{Proceedings of the ACM Web Conference 2022 (WWW '22), April 25--29, 2022, Virtual Event, Lyon, France}
\acmPrice{15.00}
\acmDOI{10.1145/3485447.3511955}
\acmISBN{978-1-4503-9096-5/22/04}

\begin{document}
\title{Efficient Neural Ranking using Forward Indexes}

\author{Jurek Leonhardt}
\affiliation{
  \institution{L3S Research Center}
  \city{Hannover}
  \country{Germany}
}
\email{leonhardt@L3S.de}

\author{Koustav Rudra}
\affiliation{
  \institution{Indian Institute of Technology\\(Indian School of Mines)}
  \city{Dhanbad}
  \country{India}
}
\email{koustav@iitism.ac.in}
\authornote{Research was primarily conducted while affiliated to L3S Research Center.}

\author{Megha Khosla}
\affiliation{
  \institution{L3S Research Center}
  \city{Hannover}
  \country{Germany}
}
\email{khosla@L3S.de}

\author{Abhijit Anand}
\affiliation{
  \institution{L3S Research Center}
  \city{Hannover}
  \country{Germany}
}
\email{aanand@L3S.de}

\author{Avishek Anand}
\affiliation{
  \institution{L3S Research Center}
  \city{Hannover}
  \country{Germany}
}
\email{anand@L3S.de}

\begin{abstract}
Neural document ranking approaches, specifically transformer models, have achieved impressive gains in ranking performance.
However, query processing using such over-parameterized models is both resource and time intensive.
In this paper, we propose the \fastforward{} index -- a simple vector forward index that facilitates ranking documents using interpolation of lexical and semantic scores -- as a replacement for contextual re-rankers and dense indexes based on nearest neighbor search.
\fastforward{} indexes rely on efficient sparse models for retrieval and merely look up pre-computed dense transformer-based vector representations of documents and passages in constant time for fast CPU-based semantic similarity computation during query processing.
We propose index pruning and theoretically grounded early stopping techniques to improve the query processing throughput.
We conduct extensive large-scale experiments on TREC-DL datasets and show improvements over hybrid indexes in performance and query processing efficiency using only CPUs.
\fastforward{} indexes can provide superior ranking performance using interpolation due to the complementary benefits of lexical and semantic similarities.
\end{abstract}

\ccsdesc[500]{Information systems~Retrieval models and ranking}

\keywords{retrieval, dense, sparse, ranking, interpolation}

\maketitle
\section{Introduction}
\label{sec:intro}
The standard approach for ad-hoc document ranking employs a retrieval phase for fast, high-recall candidate selection and a more expensive re-ranking phase.
Recent approaches have focused heavily on neural transformer-based models for both retrieval~\cite{lin2019neural,khattab2020colbert,dai_first_webconf_2020} and re-ranking~\cite{dai_sigir_2019,macavaney2019cedr,lin_emnlp_2019,hofstatter2020interpretable,hofstatter2021efficiently,li2020parade}.
However, limited work has been done in the context of efficient end-to-end solutions for ranking long documents.

There are challenges in both retrieval and re-ranking of long documents. 
The predominant strategy for the retrieval stage is based on \emph{term-based} or \emph{lexical} matching. 
Towards efficient retrieval, inverted indexes, referred to as \emph{sparse indexes}, have been employed as work horses in traditional information retrieval, exploiting sparsity in the term space for pruning out a large number of irrelevant documents.
Sparse indexes are known to suffer from the vocabulary mismatch problem, causing them to rank semantically similar documents low.
To alleviate this, \emph{dense indexes} have been proposed recently in the context of passage retrieval.
However, we find that recall considerably suffers when retrieving longer documents, i.e.\ when there are multiple vectors per document.
For example, in one case the recall for retrieving 1000 documents using a dense index is 0.58, compared to 0.70 for a sparse index (refer to Table~\ref{tab:model_eval}).

For the re-ranking stage, cross-attention models are preferred due to their ability to estimate semantic similarity, but these models are computationally expensive.
Consequently, to keep overall end-to-end ranking costs manageable, either a smaller number of documents is considered in the re-ranking phase or leaner models with fewer parameters are used, resulting in reduced ranking performance.

The aim of this paper is to propose an efficient end-to-end approach for ranking long documents without compromising effectiveness.
Firstly, we observe that dual-encoder models can be used to effectively compute the semantic similarity of a document given a query (cf.\ Table~\ref{table:interpolation_vs_reranking}).
Based on this observation, we propose a forward indexing approach that ranks documents based on interpolation, scoring them using a linear combination of the semantic similarity scores (from the re-ranking phase) and lexical matching scores (from the retrieval phase) with respect to the query~\cite{lin_emnlp_2019}.

Our proposed \emph{vector forward index} of documents, referred to as the \fastforward{} index, contains a list of pre-computed vectors for each document corresponding to its constituent passages.
Query processing using a \fastforward{} index entails retrieving documents using a sparse index, computing the semantic similarities of the retrieved documents through a sequence of index look-ups and dot products and interpolating the scores.
As a result, \fastforward{} indexes combine the benefits of both sparse and dense models, while at the same time eliminating the need for expensive nearest neighbor search (as used in dense retrieval) or contextual, GPU-accelerated re-ranking. Consequently, our work is complementary to and can be combined with other techniques that aim to improve retrieval, e.g.\ \textsc{CLEAR}~\cite{gao2020complementing} or \textsc{docT5query}~\cite{nogueira2019docT5query}.
Further, this allows us to re-rank a much higher number of documents per query, alleviating the aforementioned vocabulary mismatch problem of the sparse retriever.

Our second observation is that query processing cost in the re-ranking phase is dominated by dot-product computations between the query and the passage vectors.
Towards this, we propose two techniques to improve efficiency in query processing using \fastforward{} indexes -- \emph{sequential coalescing} and \emph{early stopping}.
In sequential coalescing, we substantially reduce the number of vectors per document by combining representations corresponding to adjacent yet similar passages.
This not only improves the query processing efficiency, but also reduces the memory footprint of the index.
Early stopping exploits the natural ordering of the sparse scores to avoid unnecessary index look-ups by maintaining a maximum score estimate for the dense scores.

We perform extensive experimental evaluation to show the performance benefits of \fastforward{} indexes and our query processing techniques. We find that interpolation using dual-encoder models consistently yields better performance than standard re-ranking using the same models. Further, increasing the sparse retrieval depth prior to interpolation improves the final ranking. Finally, we show how optimized \fastforward{} indexes accelerate the ranking, substantially decreasing the query processing time compared to hybrid retrieval, while maintaining comparable performance. To sum up, we make the following contributions:
\begin{itemize}
    \item We propose \fastforward{} indexes as an efficient approach for ad-hoc document ranking tasks.
    \item We propose novel query processing algorithms -- sequential coalescing and early stopping -- that further improve the overall query processing throughput.
    \item We perform extensive experimental evaluation to establish strong efficiency gains due to our forward indexes and query processing techniques.
\end{itemize}

\section{Related Work}
\label{sec:related}
Classical ranking approaches, such as \bm{} or the query likelihood model~\cite{Lavrenko_2001}, rely on the inverted index for efficient first-stage retrieval that stores term-level statistics like term-frequency, inverse document frequency and positional information.
We refer to this style of retrieval as sparse retrieval, since it assumes sparse document representations.
Recently, \citet{dai_first_webconf_2020, dai_context_passage_2020} proposed \deepct{}, which stores contextualized scores for terms in the inverted index for both passage and document retrieval. 
Similarly, \textsc{DeepImpact}~\cite{mallia2021learning} enriches the document collection with expansion terms to learn improved term impacts.
\textsc{Splade}~\cite{formal2021splade} aims to enrich sparse document representations using a trained contextual transformer model and sparsity regularization on the term weights.
\textsc{TILDE}~\cite{zhuang2021tilde} ranks documents using a deep query and document likelihood model.
In this work, we use the vanilla inverted index with standard term statistics for first-stage retrieval.

An alternative design strategy is to use dual-encoders to learn dense vector representations for passages or documents using contextual models~\cite{karpukhin2020dense,khattab2020colbert,gao2021coil}.
The dense vectors are then indexed in an offline phase~\cite{johnson2019billion:faiss}, where retrieval is akin to performing an approximate nearest neighbor (ANN) search given a vectorized query.
Consequently, there has been a large number of follow-up works that boost the performance of dual-encoder models by improving pre-training~\cite{chang2020pre}, optimization techniques~\cite{gao2020complementing} and negative sampling~\cite{prakash_dense_retrieval_2021,xiong2021approximate}.
In this work, we use dual-encoders for computing semantic similarity between queries and passages. 
Some approaches have also proposed architectural modifications to the dual-encoder models by having lightweight aggregations between the query and passage embeddings~\cite{chen2021co,jang2021uhd,hofstatter2021efficiently,lin2020distilling,lin2021pyserini,zhan2021optimizing,xiong2021approximate}, showing promising performance over standard term-based retrieval strategies.
\citet{nogueira2019docT5query} proposed a simple document expansion model. Note that these approaches are complementary to our work, as they can be combined with \fastforward{} indexes.

\paragraph{Models for Semantic Similarity}
While lexical matching models are typically employed in the first-stage retrieval and are known to achieve high recall, the ability to accurately determine semantic similarity is essential in the subsequent more involved and computationally expensive re-ranking stage to alleviate the vocabulary mismatch problem~\cite{mitra2019incorporating,dai2019evaluation,dai_context_passage_2020,macavaney2020expansion,mackenzie2020efficiency}. 
Computing the semantic similarity of a document given a query has been heavily researched in IR using smoothing methods~\cite{lafferty2001document}, topic models~\cite{wei2006lda}, embeddings~\cite{mitra2016dual}, personalized models~\cite{luxenburger2008matching} and other techniques.
In these classical approaches, ranking is performed by interpolating the semantic similarity scores with the lexical matching scores from the first-stage retrieval. 
Recent approaches have been dominated by over-parameterized contextual models used predominantly in re-ranking~\cite{dai_sigir_2019,macavaney2019cedr,lin_emnlp_2019,hofstatter2020interpretable,hofstatter2021efficiently,li2020parade}.
Unlike dual-encoder models used in dense retrieval, most of the above ranking models take as input a concatenation of the query and document.
This combined input results in higher query processing times for large retrieval depths since each document has to be processed in conjugation with the query string. 
Another key limitation of using contextual models for document ranking is the maximum acceptable number of input tokens for transformer models. 
Some strategies address this length limitation by document truncation~\cite{macavaney2019cedr} and chunking into passages~\cite{dai_sigir_2019,rudra_distant_cikm_2020}. However, the performance of chunking-based strategies depends on the chunking properties, i.e.\ passage length or overlap among consecutive passages~\cite{rudra_passage_chunking_2021}. 
Recent proposals include a two-stage approach, where a query-specific summary is generated by selecting relevant parts of the document, followed by re-ranking strategies over the query and summarized document~\cite{li_keyblock_2021,hofstatter_cascade_2021}.

\paragraph{Interpolation-based Rankers}
Unlike classical methods, where score interpolation is the norm, semantic similarity using a contextual model is not consistently combined with the matching score.
Recently, \citet{wang_interpolation_2021} showed that the interpolation of BERT-based retrievers and sparse retrieval methods can boost the performance.
Further, they analyzed the role of interpolation in BERT-based dense retrieval strategies (\ance{}, \textsc{RepBERT}) and found that dense retrieval alone is not enough, but interpolation with \bm{} scores is necessary. 

\section{Interpolation-Based Re-ranking}
\label{sec:interpolated_reranking}
In this section we formally introduce the problem of re-ranking. We further compare standard and interpolation-based re-ranking.

\subsection{Problem Statement}
\label{sec:problem}
The retrieval of documents or passages given a query typically happens in two stages: In the first stage, a term-frequency-based (\textbf{sparse}) retrieval method retrieves a set of documents from a large corpus. In the second stage, another model, which is usually much more computationally expensive, \textbf{re-ranks} the retrieved documents again. The re-ranking step is deemed very important for tasks that require high performance for small retrieval depths, such as question answering.

In \textbf{sparse retrieval}, we denote the top-$k_S$ documents retrieved from the sparse index for a query $q$ as $K^q_S$. The sparse score of a query-document pair $(q, d)$ is denoted by $\phi_S(q, d)$. For the \textbf{re-ranking} part, we focus on self-attention models (such as BERT) in this work. These models operate by creating (internal) high-dimensional dense representations of queries and documents, focusing on their semantic structure. We refer to the outputs of these models as \textbf{dense} or \textbf{semantic} scores and denote them by $\phi_D(q, d)$. Due to the quadratic time complexity of self-attention w.r.t.\ the document length, long documents are often split into passages, and the score of a document is then computed as the maximum of its passage scores:
\begin{equation}
    \label{eqn:maxp}
    \phi_D(q, d) = \max_{p_i \in d} {\phi_D(q, p_i)}
\end{equation}
This approach is referred to as \emph{maxP}~\cite{dai_sigir_2019}.

The retrieval approach for a query $q$ starts by retrieving $K^q_S$ from the sparse index. For each retrieved document $d \in K^q_S$, the corresponding dense score $\phi_D(q, d)$ is computed. This dense score may then be used to re-rank the retrieved set to obtain the final ranking. However, recently it has been shown that the scores of the sparse retriever, $\phi_S$, can be beneficial for re-ranking as well~\cite{lin_emnlp_2019}. To that end, an interpolation approach is employed, where the final score of a query-document pair is computed as follows:
\begin{equation}
    \label{eqn:interpolation}
    \phi(q, d) = \alpha \cdot \phi_S(q, d) + (1 - \alpha) \cdot \phi_D(q, d)
\end{equation}
Setting $\alpha = 0$ recovers the standard re-ranking procedure.

Since the set of documents retrieved by the sparse model is typically large (e.g.\ $k_S = 1000$), computing the dense score for each query-document pair can be very computationally expensive. In this paper, we focus on efficient implementations of interpolation-based re-ranking, specifically the computation of the dense scores $\phi_D$.

\subsection{Re-Ranking vs. Interpolation}
\label{sec:reranking_vs_interpolation}
\begin{table}
    \center
    \begin{tabular}{lccc}
        \toprule
                    & \trecdldf     & \trecdlds     & \trecdlp \\
        \cmidrule(lr){1-4}
        \bf \reranking \\
        \tct        & 0.685         & 0.617         & 0.694 \\
        \cls        & 0.520         & 0.522         & 0.578 \\
        \midrule
        \bf \interpolatedreranking \\
        \tct        & \bf 0.696     & \bf 0.637     & \bf 0.708 \\
        \cls        & 0.612         & 0.626         & 0.617 \\
        \bottomrule
    \end{tabular}
    \caption{Performance comparison (nDCG@$10$) between re-ranking and interpolation at retrieval depth $k_S = 1000$.}
    \label{table:interpolation_vs_reranking}
\end{table}
In order to compare traditional re-ranking with interpolation (cf. Section~\ref{sec:problem}), we retrieve $k_S = 1000$ documents from a sparse \bm{} index and compute the corresponding semantic scores for each of the query-document pairs using dense methods. We then re-rank the documents with and without interpolation. The results (cf. Table~\ref{table:interpolation_vs_reranking}) clearly show that interpolation improves over the purely semantic scoring in both document and passage re-ranking. This confirms that the lexical scores indeed hold complementary relevance signals in comparison to semantic scores. We also observe that the dual-encoder approach has better re-ranking performance.

\section{Efficient Interpolation}
\label{sec:efficient_interpolation}
\begin{figure*}
    \begin{subfigure}{.29\linewidth}
      \centering
      \includegraphics[height=0.15\textheight]{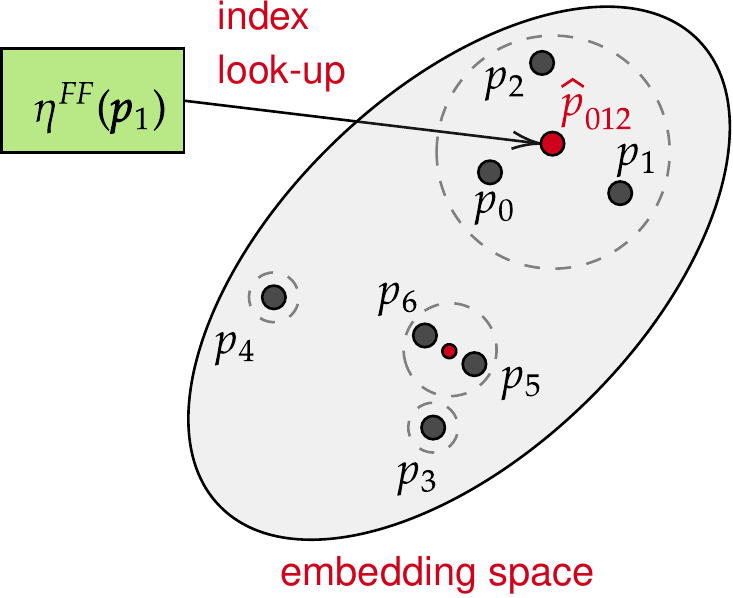}
      \caption{Sequential coalescing of maxP indexes.}
      \label{fig:coalescing}
    \end{subfigure}
    \begin{subfigure}{.7\linewidth}
      \centering
      \includegraphics[height=0.15\textheight]{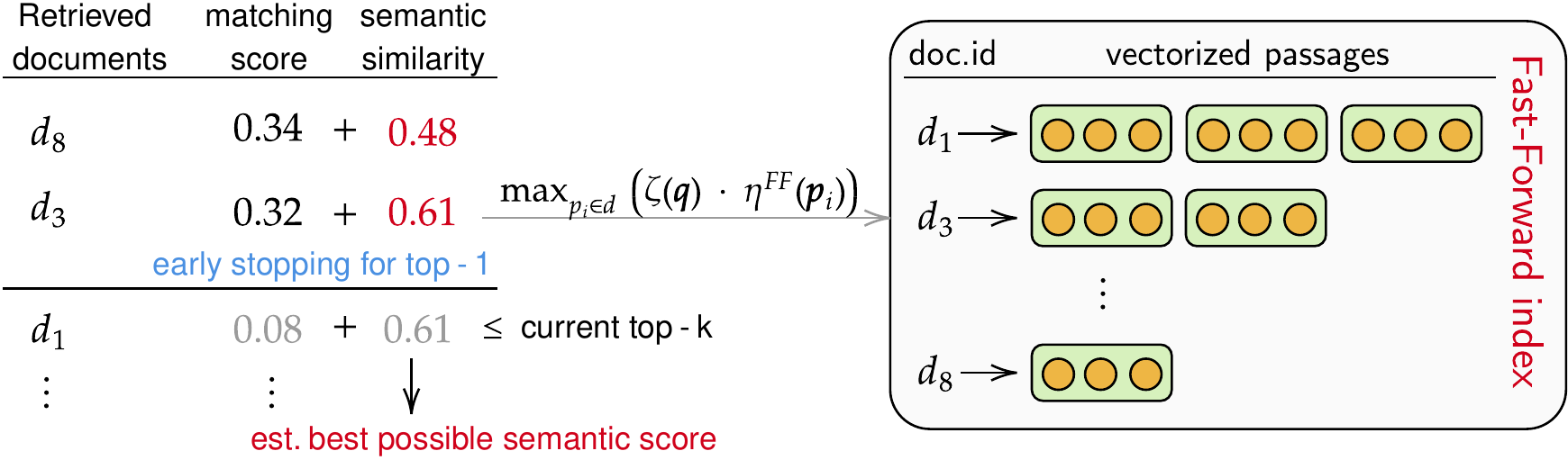}
      \caption{Early stopping during interpolation.}
      \label{fig:early_stopping}
    \end{subfigure}
    \caption{The optimization techniques that can be applied to \fastforward{} indexes. Sequential coalescing combines the representations of similar consecutive passages as their average. Note that $p_3$ and $p_5$ are not combined, as they are not consecutive passages. Early stopping reduces the number of interpolation steps by computing an approximate upper bound for the dense scores. This example depicts the most extreme case, where only the top-$1$ document is required.}
    \label{fig:ff_indexes}
\end{figure*}
Interpolation-based re-ranking, as described in Section~\ref{sec:interpolated_reranking}, is known to improve performance when applied to the results of a first-stage (sparse) retrieval step. However, the computation of $\phi_D$ (cf.\ Equation~\eqref{eqn:interpolation}) can be very slow, where cross-attention models~\cite{rudra_distant_cikm_2020,dai_sigir_2019} are more expensive than dual-encoder-based ranking strategies~\cite{lin2020distilling,chang2020pre}. In this section we propose several means of implementing interpolation-based re-ranking more efficiently.

\subsection{Hybrid Retrieval}
\label{sec:hybrid}
Hybrid retrieval is similar to standard interpolation-based re-ranking (cf. Section~\ref{sec:interpolated_reranking}). The key difference is that the dense scores $\phi_D(q, d)$ are not computed for all query-document pairs. Instead, this approach operates under the assumption that $\phi_D$ is a \textbf{dense retrieval model}, which retrieves documents $d_i$ and their scores $\phi(q, d_i)$ using nearest neighbor search given a query $q$. A hybrid retriever combines the retrieved sets of a sparse and a dense retriever.

For a query $q$, we retrieve two sets of documents, $K^q_S$ and $K^q_D$, using the sparse and dense retriever, respectively. Note that the two retrieved sets are usually not equal.
One strategy proposed in~\cite{lin2020distilling} ranks all documents in $K^q_S \cup K^q_D$, approximating missing scores. In our experiments, however, we found that \textbf{only} considering documents from $K^q_S$ for the final ranking and discarding the rest works well. The final score is thus computed as follows:
\begin{equation}
    \label{eqn:hybrid}
    \phi(q, d) = \alpha \cdot \phi_S(q, d) + (1 - \alpha) \cdot
    \begin{cases}
        \phi_D(q, d)    & d \in K^q_D \\
        \phi_S(q, d)    & d \notin K^q_D
    \end{cases}
\end{equation}
The re-ranking step in hybrid retrieval is essentially a sorting operation over the interpolated scores and takes negligible time in comparison to standard re-ranking.

\subsection{\fastforward{} Indexes}
\label{sec:forward_indexes}
The hybrid approach described in Section~\ref{sec:hybrid} has two distinct disadvantages. Firstly, in order to retrieve $K_D^q$, an (approximate) nearest neighbor search has to be performed, which is time consuming. Secondly, some of the query-document scores are missed, leading to an incomplete interpolation.

In this section we propose \fastforward{} indexes as an efficient way of computing dense scores for known documents that alleviates the aforementioned issues. Specifically, \fastforward{} indexes build upon two-tower dense retrieval models that compute the score of a query-document pair as a dot product
\begin{equation}
    \phi_D(q, d) = \zeta(q) \cdot \eta(d),
\end{equation}
where $\zeta$ and $\eta$ are the query and document encoders, respectively. Examples of such models are \ance{}~\cite{xiong2021approximate} and \tct{}~\cite{lin2020distilling}. Since the query and document representations are independent for two-tower models, we can pre-compute the document representations $\eta(d)$ for each document $d$ in the corpus. These document representations are then stored in an efficient hash map, allowing for look-ups in constant time. After the index is created, the score of a query-document pair can be computed as
\begin{equation}
    \phi_D^{FF}(q, d) = \zeta(q) \cdot \eta^{FF}(d),
\end{equation}
where the superscript $FF$ indicates the look-up of a pre-computed document representation in the \fastforward{} index. At retrieval time, only $\zeta(q)$ needs to be computed once for each query. As queries are usually short, this can be done on CPUs.

\subsection{Index Compression via Seq.\ Coalescing}
\label{sec:coalescing}
A major disadvantage of dense indexes and dense retrieval in general is the size of the final index. This is caused by two factors: Firstly, in contrast to sparse indexes, the high-dimensional dense representations can not be stored as efficiently as sparse vectors. Secondly, the dense encoders are typically transformer-based, imposing a (soft) limit on their input lengths due to their quadratic time complexity with respect to the inputs. Thus, long documents are split into passages prior to indexing (\emph{maxP} indexes).

\begin{algorithm}[t]
    \DontPrintSemicolon
    \SetKwFunction{Dist}{cosine\_distance}
    \SetKwFunction{Mean}{mean}
    \SetKw{In}{in}
    \KwIn{list of passage vectors $P$ (original order) of a document, distance threshold $\delta$}
    \KwOut{coalesced passage vectors $P'$}
    $P' \leftarrow$ empty list\;
    $\mathcal{A} \leftarrow \emptyset$\;
    \ForEach{$v$ \In $P$}{
        \uIf{first iteration}{
            \tcp{do nothing}
        }
        \uElseIf{$\Dist(v, \overline{\mathcal{A}}) \geq \delta$}{
            append $\overline{\mathcal{A}}$ to $P'$\;
            $\mathcal{A} \leftarrow \emptyset$\;
        }
        add $v$ to $\mathcal{A}$\;
        $\overline{\mathcal{A}} \leftarrow \Mean(\mathcal{A})$\;
    }
    append $\overline{\mathcal{A}}$ to $P'$\;
    \Return{$P'$}\;
    \caption{Compression of dense maxP indexes by sequential coalescing}
    \label{alg:coalescing}
\end{algorithm}
As an increase in the index size has a negative effect on retrieval latency, both for nearest neighbor search and \fastforward{} indexing as used by our approach, we exploit a \emph{sequential coalescing} approach as a way of dynamically combining the representations of consecutive passages within a single document in maxP indexes. The idea is to reduce the number of passage representations in the index for a single document. This is achieved by exploiting the \emph{topical locality} that is inherent to documents~\cite{leonhardt2020boilerplate}. For example, a single document might contain information regarding multiple topics; due to the way human readers naturally ingest information, we expect documents to be authored such that a single topic appears mostly in consecutive passages, rather than spread throughout the whole document. Our approach aims to combine consecutive passage representations that encode similar information. To that end, we employ the cosine distance function and a \emph{threshold} parameter $\delta$ that controls the degree of coalescing. Within a single document, we iterate over its passage vectors in their original order and maintain a set $\mathcal{A}$, which contains the representations of the already processed passages, and continuously compute $\overline{\mathcal{A}}$ as the average of all vectors in $\mathcal{A}$. For each new passage vector $v$, we compute its cosine distance to $\overline{\mathcal{A}}$. If it exceeds the distance threshold $\delta$, the current passages in $\mathcal{A}$ are combined as their average representation $\overline{\mathcal{A}}$. Afterwards, the combined passages are removed from $\mathcal{A}$ and $\overline{\mathcal{A}}$ is recomputed. This approach is illustrated in Algorithm~\ref{alg:coalescing}. Figure~\ref{fig:coalescing} shows an example index after coalescing. To the best of our knowledge, there are no other forward index compression techniques so far.

\subsection{Faster Interpolation by Early Stopping}
\label{sec:early_stopping}
As described in Section~\ref{sec:interpolated_reranking}, by interpolating the scores of sparse and dense retrieval models we perform implicit re-ranking, where the dense representations are pre-computed and can be looked up in a \fastforward{} index at retrieval time. Further, increasing the sparse retrieval depth $k_S$, such that $k_S > k$, where $k$ is the final number of documents, improves the performance. A drawback of this is that an increase in the number of retrieved documents also results in an increase in the number of index look-ups.

\begin{algorithm}[t]
    \DontPrintSemicolon
    \SetKwFunction{Sparse}{sparse}
    \SetKwFunction{Max}{max}
    \SetKw{In}{in}
    \SetKw{Break}{break}
    \KwIn{query $q$, sparse retrieval depth $k_S$, cut-off depth $k$, interpolation parameter $\alpha$}
    \KwOut{approximated top-$k$ scores $Q$}
    $Q \leftarrow$ priority queue of size $k$\;
    $s_{D} \leftarrow -\infty$\;
    $s_{min} \leftarrow -\infty$\;
    \ForEach{$d$ \In $\Sparse(q, k_S)$}{
        \uIf{$Q$ is full}{
            $s_{min} \leftarrow$ remove smallest item from Q\;
            $s_{best} \leftarrow \alpha \cdot \phi_S(q, d) + (1 - \alpha) \cdot s_{D}$\; \label{alg:early_stopping:sbest}
            
            \uIf{$s_{best} \leq s_{min}$}{
                \tcp{early stopping}
                put $s_{min}$ into $Q$\;
                \Break\;
            }
        }
        \tcp{approximate max. dense score}
        $s_{D} \leftarrow \Max(\phi_D(q, d), s_{D})$\;
        $s \leftarrow \alpha \cdot \phi_S(q, d) + (1 - \alpha) \cdot \phi_D(q, d)$\;
        put $\Max(s, s_{min})$ into $Q$\;
    }
    \Return{$Q$}\;
    \caption{Interpolation with early stopping}
    \label{alg:early_stopping}
\end{algorithm}
In this section we propose an extension to \fastforward{} indexes that allows for \emph{early stopping}, i.e.\ avoiding a number of unnecessary look-ups, for cases where $k_S > k$ by approximating the maximum possible dense score. The early stopping approach takes advantage of the fact that documents are ordered by their sparse scores $\phi_S(q, d)$. Since the number of retrieved documents, $k_S$, is finite, there exists an upper limit $s_D$ for the corresponding dense scores such that $\phi_D(q, d) \leq s_D \forall d \in K^q_S$. Since the retrieved documents $K^q_S$ are ordered by their sparse scores, we can simultaneously perform interpolation and re-ranking by iterating over the ordered list of documents: Let $d_i$ be the $i$-th highest ranked document by the sparse retriever. Recall that we compute the final score as follows:
\begin{equation}
    \phi(q, d_i) = \alpha \cdot \phi_S(q, d_i) + (1 - \alpha) \cdot \phi_D(q, d_i)
\end{equation}
If $i > k$, we can compute the upper bound for $\phi(q, d_i)$ by exploiting the aforementioned ordering:
\begin{equation}
    s_{best} = \alpha \cdot \phi_S(q, d_{i-1}) + (1 - \alpha) \cdot s_D
\end{equation}
In turn, this allows us to stop the interpolation and re-ranking if $s_{best} \leq s_{min}$, where $s_{min}$ denotes the score of the $k$-th document in the current ranking (i.e. the currently lowest ranked document). Intuitively, this means that we stop the computation once the \emph{highest possible} interpolated score $\phi(q, d_i)$ is too low to make a difference. The approach is illustrated in Algorithm~\ref{alg:early_stopping} and Figure~\ref{fig:early_stopping}. Since the dense scores $\phi_D$ are usually unnormalized, the upper limit $s_D$ is unknown in practice. We thus approximate it by using the highest observed dense score at any given step.

\subsubsection{Theoretical Analysis}
We first show that the early stopping criteria, when using the true maximum of the dense scores, is sufficient to obtain the top-$k$ scores.
\begin{theorem}
    Let $s_D$, as used in Algorithm~\ref{alg:early_stopping}, be the true maximum of the dense scores. Then the returned scores are the actual top-$k$ scores.
\end{theorem}
\begin{proof}
    First, note that the sparse scores, $\phi_S(q, d_i)$, are already sorted in decreasing order for a given query. By construction, the priority queue $Q$ always contains the highest scores corresponding to the list parsed so far. Let, after parsing $k$ scores, $Q$ be full. Now the possible best score $s_{best}$ is computed using the sparse score found next in the decreasing sequence and the maximum of all dense scores, $s_D$ (cf. line~\ref{alg:early_stopping:sbest}). If $s_{best}$ is less than the minimum of the scores in $Q$, then $Q$ already contains the top-$k$ scores. To see this, note that the first component of $s_{best}$ is the largest among all unseen sparse scores (as the list is sorted) and $s_D$ is maximum of the dense scores by our assumption.
\end{proof}
Next, we show that a good approximation of the top-$k$ scores can be achieved by using the sample maximum. To prove our claim, we use the Dvoretzky–Kiefer–Wolfowitz (DKW)~\cite{massart1990tight} inequality.
\begin{lemma}
    \label{lem:dkw}
    Let $X_1, X_2, ..., X_n$ be $n$ real-valued independent and identically distributed random variables with the cumulative distribution function $F(\cdot)$. Let $F_n(\cdot)$ denote the empirical cumulative distributive function, i.e.
    \begin{equation}
        F_{n}(x)=\frac{1}{n} \sum_{i=1}^{n} \mathbbm{1}_{\left\{X_{i} \leq x\right\}}, \quad x \in \mathbb{R}.
    \end{equation}
    According to the DKW inequality, the following estimate holds:
    \begin{equation}
        \Pr \left( \sup_{x \in \mathbb{R}} \left( F_{n}(x) - F(x) \right) > \epsilon \right) \leq e^{-2n \epsilon^2} \forall \epsilon \geq \sqrt{\frac{1}{2n} \ln 2}.
    \end{equation}
\end{lemma}
In the following we show that, if $s_D$ is chosen as the maximum of a large random sample drawn from the set of dense scores, then the probability that any given dense score, chosen independently and uniformly at random from the dense scores, is greater than $s_D$ is exponentially small in the sample size.
\begin{theorem}
    Let $x_1, x_2, ..., x_n$ be a real-valued independent and identically distributed random sample drawn from the distribution of the dense scores with the cumulative distribution function $F(\cdot)$. Let $z = \max{(x_1, x_2, ..., x_n)}$. Then, for every $\epsilon > \frac{1}{\sqrt{2n}} \ln 2$, we obtain
    \begin{equation}
        \label{eq:dkw_dense}
        \Pr(F(z) < 1 - \epsilon) \leq e^{-2n \epsilon^2}.
    \end{equation}
\end{theorem}
\begin{proof}
    Let $F_n(\cdot)$ denote the empirical cumulative distribution function as above. Specifically, $F_n(x)$ is equal to the fraction of variables less than or equal to $x$. We then have $F_n(z) = 1$. By Lemma~\ref{lem:dkw}, we infer
    \begin{equation}
        \Pr(F_n(z) - F(z) > \epsilon) \leq e^{-2n \epsilon^2}.
    \end{equation} 
    Substituting $F_n(z) = 1$, we obtain Equation~\eqref{eq:dkw_dense}.
\end{proof}
This implies that the probability of any random variable $X$, chosen randomly from the set of dense scores, being less than or equal to $s_D$, is greater than or equal to $1 - \epsilon$ with high probability, i.e.
\begin{equation}
    \Pr(P_D(X \leq s_D) \geq 1 - \epsilon) \geq 1 - e^{-2n \epsilon^2},
\end{equation}
where $P_D$ denotes the probability distribution of the dense scores. This means that, as our sample size grows until it reaches $k$, the approximation improves.
Note that, in our case, the dense scores are sorted (by corresponding sparse score) and thus the i.i.d.\ assumption can not be ensured. However, we observed that the dense scores are positively correlated with the sparse scores. We argue that, due to this correlation, we can approximate the maximum score well.

\section{Evaluation Setup}
\label{sec:experiments}
We consider the following baselines:

\textbf{Lexical or sparse retrievers} rely on term-based matching between queries and documents. We consider BM25, which uses term-based retrieval signals, and \deepct{}~\cite{dai_first_webconf_2020}, which is similar to BM25, but the term weights are learned in a contextualized fashion.

\textbf{Semantic or dense retrievers} retrieve documents that are semantically similar to the query in a common embedding space. We consider \tct{}~\cite{lin2020distilling} and \ance{}~\cite{xiong2021approximate}. Both approaches are based on BERT encoders. Large documents are split into passages before indexing (maxP). These dense retrievers use exact (brute-force) nearest neighbor search as opposed to approximate nearest neighbor (ANN) search.

\textbf{Hybrid retrievers} interpolate sparse and dense retriever scores. We consider \clear{}~\cite{gao2020complementing}, a retrieval model that complements lexical models with semantic matching. Additionally, we consider the hybrid strategy described in Section~\ref{sec:hybrid} as a baseline, using the dense retrievers above.

\textbf{Contextual dense re-rankers} operate on the documents retrieved by a sparse retriever (\bm{}). Each query-document pair is input into the re-ranker, which outputs a corresponding score. In this paper, we use a \cls{} re-ranker, where the output corresponding to the classification token is used as the score. Note that re-ranking is performed using the full documents (i.e.\ documents are not split into passages). If an input exceeds $512$ tokens, it is truncated.

\paragraph{Datasets and Hyperparameters}
We conduct experiments on three datasets from the TREC Deep Learning track, \trecdldf{}, \trecdlds{} and \trecdlp{}, to evaluate the effectiveness and efficiency of retrieval and re-ranking strategies on the MS MARCO collection. Each test set has a total of 200 queries. We use the \pyserini{} toolkit~\cite{lin2021pyserini} for our retrieval experiments and the MS MARCO development set to determine $\alpha = 0.2$ for \tct{}, $\alpha = 0.5$ for \ance{} and $\alpha = 0.7$ for \cls{}. Latency is computed as the sum of scoring, interpolation and sorting cost. Tokenization cost is ignored. We report the average processing time per query in the test set. Where applicable, dense models use a batch size of 256. More details can be found in Appendix~\ref{app:hardware}.

\section{Experimental Results}
\label{sec:result}
In this section we perform large-scale experiments to show the effectiveness and efficiency of the proposed \fastforward{} indexes. 

\paragraph{RQ1. How does interpolation-based re-ranking with dual-encoders compare to other methods?}
\begin{table*}
    \centering
    \begin{tabular}{lccccccccc}
        \toprule
            & \multicolumn{3}{c}{\trecdldf}
            & \multicolumn{3}{c}{\trecdlds}
            & \multicolumn{3}{c}{\trecdlp} \\
            \cmidrule(lr){2-4}
            \cmidrule(lr){5-7}
            \cmidrule(lr){8-10}
            & $\text{AP}_\text{1k}$ & $\text{R}_\text{1k}$ & $\text{nDCG}_\text{10}$
            & $\text{AP}_\text{1k}$ & $\text{R}_\text{1k}$ & $\text{nDCG}_\text{10}$
            & $\text{AP}_\text{1k}$ & $\text{R}_\text{1k}$ & $\text{nDCG}_\text{10}$ \\
        \midrule
        \multicolumn{10}{l}{\bf \sparseretrieval} \\
        \bm             & 0.331     & 0.697     & $0.519^{1-3}$
                        & 0.404     & 0.809     & $0.527^{1-3}$
                        & 0.301     & 0.750     & $0.506^{1-3}$ \\
        \deepct         & -         & -         & 0.544
                        & -         & -         & -
                        & 0.422     & 0.756     & 0.551 \\
        \midrule
        \multicolumn{10}{l}{\bf \denseretrieval} \\
        \tct            & 0.279     & 0.576     & $0.612^{1}$
                        & 0.372     & 0.728     & $0.586^{1,2}$
                        & 0.391     & 0.792     & 0.670 \\
        \ance           & 0.254     & 0.510     & $0.633^{1}$
                        & 0.401     & 0.681     & 0.633
                        & 0.371     & 0.755     & 0.645 \\
        \midrule
        \multicolumn{10}{l}{\bf \hybrid} \\
        \clear          & -         & -         & -
                        & -         & -         & -
                        & 0.511     & 0.812     & 0.699 \\
        \midrule
        \multicolumn{10}{l}{\bf \reranking} \\
        \tct            & 0.370     & 0.697     & 0.685
                        & 0.414     & 0.809     & 0.617
                        & 0.423     & 0.750     & 0.694 \\
        \ance           & 0.336     & 0.697     & 0.654
                        & 0.426     & 0.809     & 0.630
                        & 0.389     & 0.750     & 0.679 \\
        \cls            & 0.283     & 0.697     & $0.520^{1-3}$
                        & 0.329     & 0.809     & $0.522^{1-3}$
                        & 0.353     & 0.750     & $0.578^{1,2}$ \\
        \midrule
        \multicolumn{10}{l}{\bf \interpolatedreranking} \\
        $\tct^1$        & 0.406     & 0.697     & 0.696
                        & 0.469     & 0.809     & 0.637
                        & 0.438     & 0.750     & 0.708 \\
        $\ance^2$       & 0.387     & 0.697     & 0.673
                        & 0.490     & 0.809     & 0.655
                        & 0.417     & 0.750     & 0.680 \\
        $\cls^3$        & 0.365     & 0.697     & 0.612
                        & 0.460     & 0.809     & 0.626
                        & 0.378     & 0.750     & 0.617 \\
        \bottomrule
    \end{tabular}
    \caption{Retrieval performance. Retrievers use depths $k_S = 1000$ (sparse) and $k_D = 10000$ (dense). Dense retrievers retrieve passages and perform maxP aggregation for documents. Scores for \clear{} and \deepct{} are taken from the corresponding papers~\cite{gao2020complementing,gao2021coil}. Superscripts indicate statistically significant improvements using two-paired tests with a sig.\ level of $95\%$~\cite{paired_significance_test}.}
    \label{tab:model_eval}
\end{table*}
In Table~\ref{tab:model_eval}, we report the performance of sparse, dense and hybrid retrievers, re-rankers and interpolation.

First, we observe that dense retrieval strategies perform better than sparse ones in terms of nDCG, but have poor recall except on \trecdlp{}. The contextual weights learned by \deepct{} are better than tf-idf based retrieval (\bm{}), but fall short of dense semantic retrieval strategies (\tct{} and \ance{}). However, the overlap among retrieved documents is rather low, reflecting that dense retrieval cannot match query and document terms well.

Second, dual-encoder-based (\tct{} and \ance{}) perform better than contextual (\cls{}) re-rankers. In this setup, we first retrieve $k_S = 1000$ documents using a sparse retriever and re-rank them. This approach benefits from high recall in the first stage and promotes the relevant documents to the top of the list through the dense semantic re-ranker. However, re-ranking is typically time-consuming and requires GPU acceleration. The improvement of \tct{} and \ance{} over \cls{} also suggests that dual-encoder-based re-ranking strategies are better than cross-interaction-based methods. However, the difference could also be attributed to the fact that \cls{} does not follow the maxP approach (cf. Section~\ref{sec:problem}).

Finally, interpolation-based re-ranking, which combines the benefits of sparse and dense scores, significantly outperforms the \cls{} re-ranker and dense retrievers. Recall that dense re-rankers operate solely based on the dense scores and discard the sparse \bm{} scores of the query-document pairs. The superiority of interpolation-based methods is also supported by evidence from recent studies~\cite{chang2020pre,chen2021co,gao2020complementing,gao2021coil}.

\paragraph{RQ2. Do \fastforward{} indexes allow for efficient interpolation at higher retrieval depths?}
\begin{table*}
    \centering
    \resizebox{\linewidth}{!}{
    \begin{tabular}{lccccccccccccc}
        \toprule
        &
        & \multicolumn{6}{c}{\trecdldf}
        & \multicolumn{6}{c}{\trecdlds}
        \\
        \cmidrule(lr){3-8}
        \cmidrule(lr){9-14}
        & \multirowcell{2}[-0.5ex]{millisec.\\per query}
        & \multicolumn{3}{c}{$k_S = 1000$}
        & \multicolumn{3}{c}{$k_S = 5000$}
        & \multicolumn{3}{c}{$k_S = 1000$}
        & \multicolumn{3}{c}{$k_S = 5000$}
        \\
        \cmidrule(lr){3-5}
        \cmidrule(lr){6-8}
        \cmidrule(lr){9-11}
        \cmidrule(lr){12-14}
        &
        & $\text{AP}_\text{1k}$ & $\text{R}_\text{1k}$ & $\text{nDCG}_\text{20}$
        & $\text{AP}_\text{1k}$ & $\text{R}_\text{1k}$ & $\text{nDCG}_\text{20}$
        & $\text{AP}_\text{1k}$ & $\text{R}_\text{1k}$ & $\text{nDCG}_\text{20}$
        & $\text{AP}_\text{1k}$ & $\text{R}_\text{1k}$ & $\text{nDCG}_\text{20}$
        \\
        \midrule
        \multicolumn{10}{l}{\bf \hybrid} \\
        \bm, \tct
        & \cpu{582}
        & 0.394         & 0.697         & 0.655 & 0.385         & 0.729 & 0.645 
        & 0.463         & 0.809         & 0.615 & 0.469         & 0.852 & 0.621 
        \\
        \bm, \ance
        & \cpu{582}
        & 0.379         & 0.697         & 0.633 & 0.373         & 0.727 & 0.628
        & 0.479         & 0.809         & 0.624 & 0.488         & 0.846 & 0.632 
        \\
        \midrule
        \multicolumn{10}{l}{\bf \reranking} \\
        \tct
        & \gpu{1189} + \cpu{2}
        & 0.370             & 0.697         & 0.632         & 0.334     & 0.703     & $0.609^{1}$ 
        & 0.414             & 0.809         & $0.587^{1}$   & 0.405     & 0.794     & $0.585^{1,3,4}$ 
        \\
        \ance
        & \gpu{1189} + \cpu{2}
        & 0.336             & 0.697         & 0.614         & 0.304     & 0.647     & 0.607 
        & 0.426             & 0.809         & $0.595^{3}$   & 0.422     & 0.761     & 0.604 
        \\
        \cls
        & \gpu{185} + \cpu{2}
        & 0.283             & 0.697         & $0.494^{1-5}$ & 0.159     & 0.559     & 0.289 
        & 0.329             & 0.809         & $0.512^{1-5}$ & 0.221     & 0.727     & $0.375^{1-5}$ 
        \\
        \midrule
        \multicolumn{10}{l}{\bf \interpolatedreranking} \\
        $\tct^1$
        & \gpu{1189} + \cpu{14}
        & \hide{0.406}   & \hide{0.697}   & \hide{0.655} & \hide{0.411}   & \hide{0.745} & \hide{0.653}
        & \hide{0.469}   & \hide{0.809}   & \hide{0.621} & \hide{0.478}   & \hide{0.838} & \hide{0.626}
        \\
        \quad \fastforward
        & \cpu{253}
        & 0.406         & 0.697     & 0.655         & 0.411     & 0.745     & 0.653 
        & 0.469         & 0.809     & 0.621         & 0.478     & 0.838     & 0.626 
        \\
        \quad \quad coalesced$^2$
        & \cpu{109}
        & 0.379         & 0.697     & 0.630         & 0.379     & 0.732     & 0.625 
        & 0.440         & 0.809     & $0.594^{1}$   & 0.447     & 0.837     & 0.607 
        \\
        $\ance^3$
        & \gpu{1189} + \cpu{14}
        & \hide{0.387}   & \hide{0.697}   & \hide{0.638} & \hide{0.393}   & \hide{0.732} & \hide{0.639} 
        & \hide{0.490}   & \hide{0.809}   & \hide{0.630} & \hide{0.502}   & \hide{0.828} & \hide{0.640} 
        \\
        \quad \fastforward
        & \cpu{253}
        & 0.387         & 0.697     & 0.638     & 0.393     & 0.732     & 0.639 
        & 0.490         & 0.809     & 0.630     & 0.502     & 0.828     & 0.640 
        \\
        \quad \quad coalesced$^4$
        & \cpu{121}
        & 0.372         & 0.697     & 0.625     & 0.375     & 0.723     & 0.628 
        & 0.471         & 0.809     & 0.622     & 0.479     & 0.823     & 0.629 
        \\
        $\cls^5$
        & \gpu{185} + \cpu{14}
        & 0.365         & 0.697     & 0.585     & 0.357     & 0.708     & 0.562 
        & 0.460         & 0.809     & 0.602     & 0.459     & 0.839     & 0.601 
        \\
        \bottomrule
    \end{tabular}}
    \caption{Document retrieval performance. Latency is reported for $k_S = 5000$ on \cpu{CPU} and \gpu{GPU}. The coalesced \fastforward{} indexes are compressed to approximately 25\% of their original size. Hybrid retrievers use a dense retrieval depth of $k_D = 1000$. Superscripts indicate statistically significant improvements using two-paired tests with a sig.\ level of $95\%$~\cite{paired_significance_test}.}
    \label{tab:model_eval_doc}
\end{table*}
\begin{table}
    \centering
    \resizebox{\linewidth}{!}{
    \begin{tabular}{lccccccc}
        \toprule
        & \multirowcell{2}[-0.5ex]{millisec.\\per query}
        & \multicolumn{2}{c}{$k_S = 1000$}
        & \multicolumn{2}{c}{$k_S = 5000$}
        \\
        \cmidrule(lr){3-4}
        \cmidrule(lr){5-6}
        &
        & $\text{AP}_\text{1k}$ & $\text{RR}_\text{10}$
        & $\text{AP}_\text{1k}$ & $\text{RR}_\text{10}$
        \\
        \midrule
        \multicolumn{6}{l}{\bf \hybrid} \\
        \bm, \tct
        & \cpu{307}
        & 0.434             & 0.894         & 0.454         & 0.902
        \\
        \bm, \ance
        & \cpu{307}
        & 0.410             & 0.856         & 0.422         & 0.864
        \\
        \midrule
        \multicolumn{6}{l}{\bf \reranking} \\
        \tct
        & \gpu{186} + \cpu{2}
        & 0.426             & 0.827         & 0.439         & 0.842
        \\
        \ance
        & \gpu{186} + \cpu{2}
        & 0.389             & 0.836         & 0.392         & 0.857
        \\
        \cls
        & \gpu{185} + \cpu{2}
        & 0.353             & 0.715         & 0.275         & 0.576
        \\
        \midrule
        \multicolumn{6}{l}{\bf \interpolatedreranking} \\
        \tct
        & \gpu{186} + \cpu{14}
        & \hide{0.438}   & \hide{0.894}   & \hide{0.460}   & \hide{0.902}
        \\
        \quad \fastforward
        & \cpu{114}
        & 0.438         & 0.894         & 0.460         & 0.902
        \\
        \quad \quad early stopping
        & \cpu{72}
        & -             & 0.894         & -             & 0.902
        \\
        \ance
        & \gpu{186} + \cpu{14}
        & \hide{0.417}   & \hide{0.856}   & \hide{0.435}   & \hide{0.864}
        \\
        \quad \fastforward
        & \cpu{114}
        & 0.417         & 0.856         & 0.435         & 0.864
        \\
        \quad \quad early stopping
        & \cpu{52}
        & -             & 0.856         & -             & 0.864
        \\
        \cls
        & \gpu{185} + \cpu{14}
        & 0.378         & 0.809         & 0.392         & 0.832
        \\
        \bottomrule
    \end{tabular}}
    \caption{Retrieval performance on \trecdlp{}. Latency is reported for $k_S = 5000$ on \cpu{CPU} and \gpu{GPU}. Hybrid retrievers use a dense retrieval depth of $k_D = 1000$.}
    \label{tab:model_eval_passage}
\end{table}
Tables~\ref{tab:model_eval_doc} and~\ref{tab:model_eval_passage} show results of re-ranking, hybrid retrieval and interpolation on document and passage datasets, respectively. The metrics are computed for two sparse retrieval depths, $k_S = 1000$ and $k_S = 5000$.

We observe that taking the sparse component into account in the score computation (as is done by the interpolation and hybrid methods) causes performance to improve with retrieval depth. Specifically, some queries receive a considerable recall boost, capturing more relevant documents with large retrieval depths. Interpolation based on \fastforward{} indexes achieves substantially lower latency compared to other methods. Pre-computing the document representations allows for fast look-ups during retrieval time. As only the query needs to be encoded by the dense model, both retrieval and re-ranking can be performed on the CPU while still offering considerable improvements in query processing time. Note that for \cls{}, the input length is limited, causing documents to be truncated, similarly to the \emph{firstP} approach. As a result, the latency is much lower, but in turn the performance suffers. It is important to note here, that, in principle, \fastforward{} indexes can also be used in combination with firstP models.

The hybrid retrieval strategy, as described in Section~\ref{sec:hybrid}, shows good performance. However, as the dense indexes require nearest neighbor search for retrieval, the query processing latency is much higher than for interpolation using \fastforward{} indexes.

Finally, dense re-rankers do not profit reliably from increased sparse retrieval depth; on the contrary, the performance drops in some cases. This trend is more apparent for the document retrieval datasets with higher values of $k_S$. We hypothesize that dense rankers only focus on semantic matching and are sensitive to topic drift, causing them to rank irrelevant documents in the top-$5000$ higher.

\paragraph{RQ3. Can the re-ranking efficiency be improved by reducing the \fastforward{} index size using sequential coalescing?}
In order to evaluate this approach, we first take the pre-trained \tct{} dense index of the MS MARCO corpus, apply sequential coalescing with varying values for $\delta$ and evaluate each resulting compressed index using the \trecdldf{} testset.
\begin{figure}
    \centering
    \includegraphics[width=\linewidth]{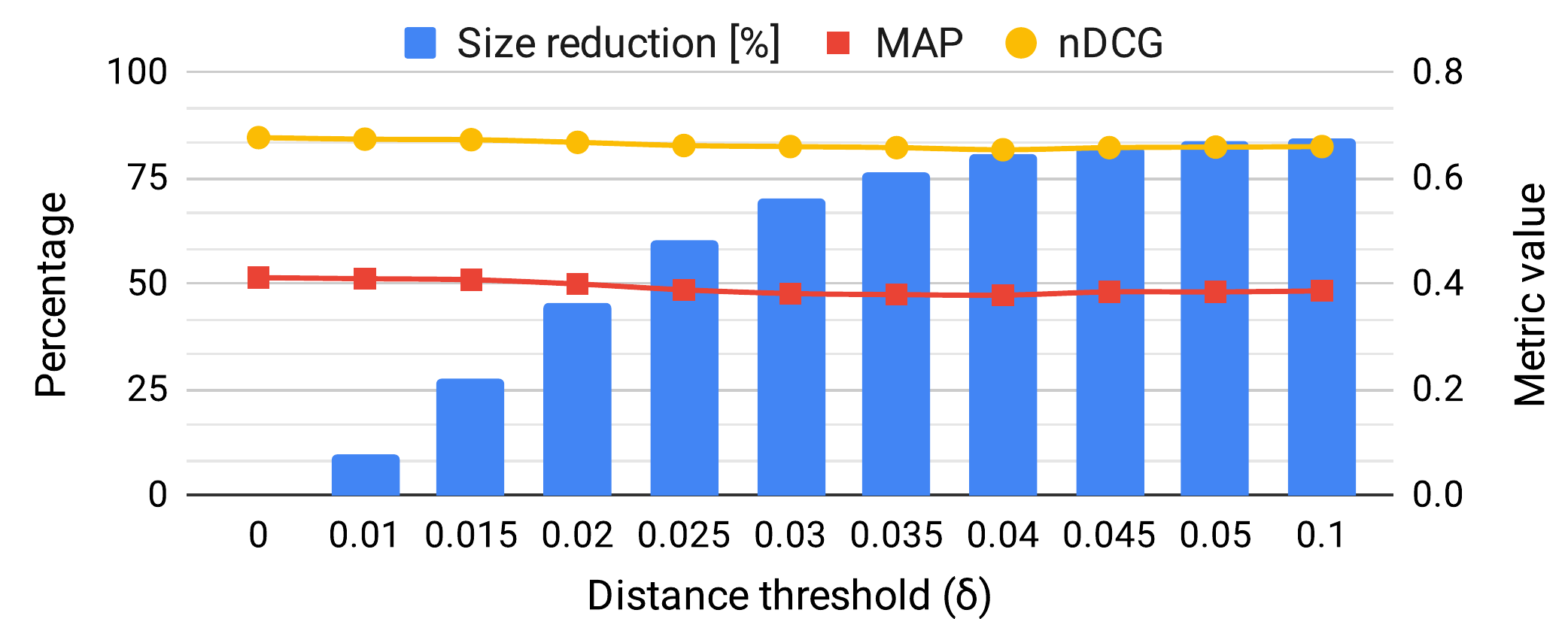}
    \caption{Sequential coalescing applied to \trecdldf{}. The plot shows the index size reduction in terms of the number of passages and the corresponding metric values for \fastforward{} interpolation with \tct{}.}
    \label{fig:coalescing_performance}
\end{figure}
The results are illustrated in Figure~\ref{fig:coalescing_performance}. It is evident that, by combining the passage representations, the number of vectors in the index can be reduced by more than 80\% in the most extreme case, where only a single vector per document remains. At the same time, the performance is correlated with the granularity of the representations. However, the drops are relatively small. For example, for $\delta = 0.025$, the index size is reduced by more than half, while the nDCG decreases by roughly $0.015$ (3\%).

Additionally, Table~\ref{tab:model_eval_doc} shows the detailed performance of coalesced \fastforward{} indexes on the document datasets. We chose the indexes corresponding to $\delta = 0.035$ (\tct{}) and $\delta = 0.003$ (\ance{}), both of which are compressed to approximately 25\% of their original size. This is reflected in the query processing latency, which is reduced by more than half. The overall performance drops to some extent, as expected, however, these drops are not statistically significant in all but one case. The trade-off between latency (index size) and performance can be controlled by varying the threshold $\delta$.

\paragraph{RQ4. Can the re-ranking efficiency be improved by limiting the number of \fastforward{} look-ups?}
\begin{figure}
    \centering
    \includegraphics[width=\linewidth]{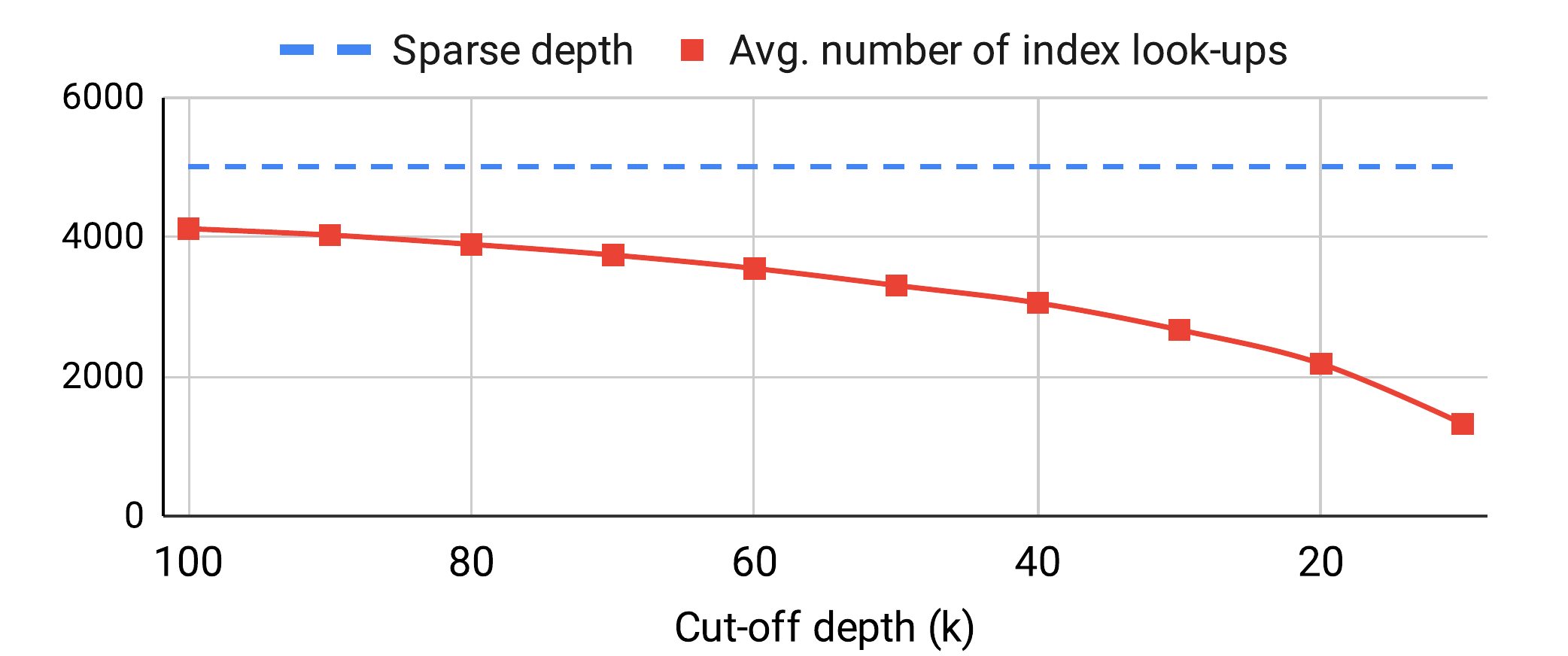}
    \caption{The average number of \fastforward{} index look-ups per query for interpolation with early stopping at varying cut-off depths $k$ on \trecdlp{} with $k_S = 5000$ using \ance{}.}
    \label{fig:early_stopping_lookups}
\end{figure}
We start by evaluating the utility of the early stopping approach described in Section~\ref{sec:early_stopping} on the \trecdlp{} dataset. Figure~\ref{fig:early_stopping_lookups} shows the average number of look-ups performed in the \fastforward{} index during interpolation w.r.t.\ the cut-off depth $k$. We observe that, for $k = 100$, early stopping already leads to a reduction of almost 20\% in the number of look-ups. Decreasing $k$ further leads to a significant reduction of look-ups, resulting in improved query processing latency. As lower cut-off depths (i.e.\ $k < 100$) are typically used in downstream tasks, such as question answering, the early stopping approach for low values of $k$ turns out to be particularly helpful.

Table~\ref{tab:model_eval_passage} shows early stopping applied to the passage dataset to retrieve the top-$10$ passages and compute reciprocal rank. It is evident that, even though the algorithm approximates the maximum dense score (cf. Section~\ref{sec:early_stopping}), the resulting performance is identical, which means that the approximation was accurate in both cases and did not incur any performance hit. Further, the query processing time is decreased by up to a half compared to standard interpolation. Note that early stopping depends on the value of $\alpha$, hence the latency varies between \tct{} and \ance{}.

\section{Conclusion}
\label{sec:conclusion}
In this paper we propose \fastforward{} indexes, a simple yet effective and efficient look-up-based interpolation method that combines document retrieval and re-ranking. \fastforward{} indexes are based on dense dual-encoder models, exploiting the fact that document representations can be pre-processed and stored, providing efficient access in constant time.
Using interpolation, we observe increased performance compared to hybrid retrieval. Further, we achieve improvements of up to 75\% in memory footprint and query processing latency due to our optimization techniques, \emph{sequential coalescing} and \emph{early stopping}. At the same time, our method solely requires CPU computations, completely eliminating the need for expensive GPU-accelerated re-ranking.

\begin{acks}
Funding for this project was in part provided by EU Horizon 2020 grant no.\ 871042 (\emph{SoBigData++}) and 832921 (\emph{MIRROR}) and BMBF grant no.\ 01DD20003 (\emph{LeibnizKILabor}).
\end{acks}

\bibliographystyle{ACM-Reference-Format}
\bibliography{references}

\clearpage
\appendix
\section{Experimental Details}
\label{app:exp_details}
In this section we provide details regarding our experiments to ensure reproducibility. This includes hardware, software, model hyperparameters as well as techniques employed.

\subsection{Hardware Configuration and Latency Measurements}
\label{app:hardware}
Our experiments are performed on a single machine using an Intel Xeon Silver 4210 CPU with 40 cores, 256GB of RAM and an NVIDIA Tesla V100 GPU. In order to measure the per-query latency numbers, we perform each experiment four times and report the average latency, excluding the first measurement (in order to account for any potential caching). In general, latency is reported as the sum of scoring (this includes operations like encoding queries and documents, obtaining representations from a \fastforward{} index, computing the scores as dot-products and so on), interpolation (cf.\ Equation~\eqref{eqn:interpolation}) and sorting cost. Any pre-processing or tokenization cost is ignored. Further, the first-stage (sparse) retrieval step is not included, as it is constant for all methods. The \fastforward{} indexes are loaded into the main memory entirely before they are accessed.

\subsection{Software and Hyperparameters}
\label{app:software}
\begin{table}
    \center
    \begin{tabular}{ll}
        \toprule
        \multicolumn{2}{c}{\ance} \\
        \midrule
        \trecdldf,                  & \texttt{castorini/ance-msmarco-doc-maxp} \\
        \trecdlds                   & \texttt{msmarco-doc-ance-maxp-bf} \\
        \midrule
        \multirowcell{2}{\trecdlp}  & \texttt{castorini/ance-msmarco-passage} \\
                                    & \texttt{msmarco-passage-ance-bf} \\
        \midrule
        \midrule
        \multicolumn{2}{c}{\tct} \\
        \midrule
        \trecdldf,                  & \texttt{castorini/tct\_colbert-msmarco} \\
        \trecdlds                   & \texttt{msmarco-doc-tct\_colbert-bf} \\
        \midrule
        \multirowcell{2}{\trecdlp}  & \texttt{castorini/tct\_colbert-msmarco} \\
                                    & \texttt{msmarco-passage-tct\_colbert-bf} \\
        \bottomrule
    \end{tabular}
    \caption{The pre-trained dense encoders and corresponding indexes we used in our experiments. In each cell, the first line corresponds to a pre-trained encoder (to be obtained from the \emph{HuggingFace Hub}) and the second line is a pre-built index provided by \pyserini{}.}
    \label{tab:encoders_indexes}
\end{table}
We use the \pyserini{} toolkit for all of our retrieval experiments, which uses the \emph{HuggingFace transformers} library internally. \pyserini{} provides a number of pre-trained encoders and corresponding indexes. Table~\ref{tab:encoders_indexes} gives an overview over the ones we used for our experiments. Our dense encoders (\ance{} and \tct{}) output 768-dimensional representations. Our sparse \bm{} retriever is provided by \pyserini{} as well. We use the pre-built indexes \texttt{msmarco-passage} ($k_1=0.82$, $b=0.68$) and \texttt{msmarco-doc} ($k_1=4.46$, $b=0.82$).

\end{document}